\documentclass[11pt]{article}

\newif\ifFull
\newif\ifShort

\Fullfalse

\ifFull
\Shortfalse
\else
\Shorttrue
\fi

\advance \topmargin by -\headheight
\advance \topmargin by -\headsep
\evensidemargin \oddsidemargin
\usepackage{hyperref}       % hyperlinks
\usepackage{url}            % simple URL typesetting
\usepackage{booktabs}       % professional-quality tables
\usepackage{amsfonts}       % blackboard math symbols
\usepackage{microtype}      % microtypography
\usepackage{xcolor}
\usepackage{amsfonts}
\usepackage{amsthm}
\usepackage{amssymb}
\usepackage{amsmath}
\usepackage{relsize}
\usepackage{xcolor}
\usepackage{algorithmic}
\usepackage{algorithm}
\usepackage{multirow}
\usepackage{wrapfig}
\usepackage{caption}
\usepackage{sidecap}
\usepackage{subcaption}
\usepackage{enumitem}
\usepackage{graphicx}

\mathchardef\mhyphen="2D

\newcommand{\pos}{\mathsf{pos}}
\newcommand{\numsamples}{k}

\newcommand{\samplesize}{m}
\newcommand{\sampleparam}{q}
\newcommand{\sampleprob}{\gamma}

\newcommand{\cursize}{\mathsf{cursize}}
\newcommand{\samplefunc}{\mathsf{sample}}
\newcommand{\samplepass}{\mathsf{samples}}

\newcommand{\getsamplepos}{\mathsf{getsamplepos}}
\newcommand{\getsamplesize}{\mathsf{getsamplesize}}

\newcommand{\sample}{\mathsf{s}}

\newcommand{\dummy}{\mathsf{dummy}}
\newcommand{\key}{\mathsf{key}}

\newcommand{\Binom}{\mathsf{Binom}}

\newcommand{\True}{\mathsf{True}}
\newcommand{\False}{\mathsf{False}}
\newcommand{\append}{\mathsf{append}}
\newcommand{\alg}{\mathcal{A}}

\newcommand{\enc}{\mathsf{enc}}
\newcommand{\reenc}{\mathsf{re\mhyphen enc}}
\newcommand{\dec}{\mathsf{dec}}
\newcommand{\nexte}{e_\mathsf{next}}
\newcommand{\samplemember}{\mathsf{samplemember}}
\newcommand{\init}{\mathsf{initialize}}

\newcommand{\mech}{\mathcal{M}}
\newcommand{\eps}{\epsilon}
\newcommand{\Renyi}{R\'{e}nyi}

\newcommand{\poisson}{\mathsf{Poisson}}
\newcommand{\shuffle}{\mathsf{Shuffle}}
\newcommand{\oshuffle}{\mathsf{oblshuffle}}

\newcommand{\cd}{\mathcal{D}}

\newcommand{\swo}{\mathsf{SWO}}

\theoremstyle{definition}

\theoremstyle{lemma}
\newtheorem{lemma}{Lemma}

\begin{document}
\title{Oblivious Sampling Algorithms for Private Data Analysis\footnotemark[1]}

\author{Sajin Sasy\footnotemark[2] \footnotemark[4]\and Olga Ohrimenko\footnotemark[3] \footnotemark[4]}
\date{}

\renewcommand{\thefootnote}{\fnsymbol{footnote}}
\footnotetext[1]{Appeared in Advances in Neural Information Processing Systems 33 (NeurIPS 2019)}
\footnotetext[2]{University of Waterloo}
\footnotetext[3]{The University of Melbourne}
\footnotetext[4]{Work done while at Microsoft}

\maketitle

\begin{abstract}

We study secure and privacy-preserving data analysis
based on queries executed on samples from a dataset.
Trusted execution environments (TEEs) can be used to
protect the content of the data during query computation,
while supporting differential-private (DP) queries in TEEs
provides record privacy when query output is revealed.
Support for sample-based queries is attractive
due to \emph{privacy amplification}
since not all dataset is used to answer a query but only a small subset.
However, extracting data samples with TEEs
while proving strong DP guarantees is not
trivial as secrecy of sample indices has to be preserved.
To this end, we design efficient secure variants of common sampling algorithms.
Experimentally we show that accuracy of models
trained with shuffling and sampling is the same for
differentially private models for MNIST and CIFAR-10,
while sampling provides stronger privacy guarantees than shuffling.

\end{abstract}

\section{Introduction}

Sensitive and proprietary datasets (e.g., health,
personal and financial records, laboratory experiments,
emails, and other personal digital communication)
often come with strong privacy and access control
requirements and regulations that are hard to maintain and guarantee
end-to-end.
The fears of data leakage may block datasets from being used by data scientists
and
prevent collaboration
and information sharing between multiple parties towards a common good (e.g.,
training a disease detection model across data from multiple hospitals).
For example,
the authors of~\cite{DBLP:journals/corr/abs-1802-08232,DBLP:conf/ccs/FredriksonJR15,DBLP:conf/sp/ShokriSSS17} show that machine learning
models can memorize individual data records,
while information not required for the agreed upon learning task may be leaked in
collaborative learning~\cite{DBLP:journals/corr/abs-1805-04049}.
To this end, we are interested in designing the following secure data query
framework:
\begin{itemize}
\item A single or multiple \emph{data owners} contribute their datasets to the platform
while expecting strong security privacy guarantees on the usage of their data;
\item The \emph{framework} acts as a gatekeeper of the data
and a computing resource of the data scientist:
it can compute queries on her behalf while ensuring that data is protected from third parties;
\item \emph{Data scientist} queries the data via the framework
via a range of queries varying from approximating sample statistics
to training complex machine learning models.
\end{itemize}
The  goal of the framework is to allow
data scientist to query the data while providing strong privacy guarantees to
data owners on their data.
The framework aims to protect against
two classes of attackers: the owner of the computing infrastructure of the framework
and the data scientist.

The data scientist may try to infer more information
about the dataset than what is available through a (restricted) class of queries
supported by the framework. We consider the following two collusion scenarios.
As the framework may be hosted in the cloud
or on premise of the data scientist's organization,
the infrastructure is not trusted
as one can access the data
without using the query interface.
The second collusion may occur
in a multi-data-owner scenario
where the data scientist could
combine the answer of a query and data
of one of the parties to infer
information about 
other parties' data.
Hence, the attacker may have auxiliary information about the data.

In the view of the above requirements and threat model
we propose \emph{Private Sampling-based Query Framework}.
It relies on secure hardware to protect data content
and restrict data access.
Additionally, it supports sample-based
 differentially private queries for efficiency and privacy.
However, naive combination of these
components does not lead to an end-to-end secure system
for the following reason.
Differential privacy guarantees
for sampling algorithms (including machine learning
model training that build on them~\cite{abadi,DBLP:conf/iclr/McMahanRT018,dpmp_dl}) are satisfied only
if the sample is hidden.
Unfortunately as we will see this is not the case with
secure hardware
due to leakage of memory
access patterns.
To this end, we design novel algorithms for producing data samples
using two common sampling techniques, Sampling without replacement and Poisson,
with the guarantee
that whoever observes data access patterns cannot identify
the indices of the dataset used in the samples.
We also argue that if privacy of data during model training is a
requirement then sampling should be used instead
of the default use of shuffling since it incurs smaller
privacy loss in return to similar accuracy as we show experimentally.
We now describe components of our \emph{Private Sampling-based Query Framework}.

\paragraph{Framework security}
In order to protect data content and computation
from the framework host,
we rely on encryption and trusted execution environments (TEE).
TEEs can be enabled using secure hardware capabilities such
as Intel SGX~\cite{sgx} which provides a set of CPU instructions that
gives access to special memory regions (enclaves) where
encrypted data is loaded, decrypted and computed on.
Importantly access to this region is restricted and data is always encrypted in memory.
One can also verify the code
and data that is loaded in TEEs via attestation.
Hence, data owners can provide data encrypted
under the secret keys that
are available only to TEEs
running specific code (e.g., differentially private algorithms).
\ifFull
Privacy of the query code is not in scope
and we assume that algorithms running in TEE are public.
See~\cite{Ohrimenko2016,prochlo,sgx,vc3} for more details.
\fi
Some of the limitations of TEEs include resource sharing
with the rest of the system (e.g., caches, memory, network),
which may lead to side-channels~\cite{Brasser:2017:SGE:3154768.3154779,sgxcacheattacks,Osvik2006}.
Another limitation of existing TEEs
is the amount of available enclave memory
(e.g., Intel Skylake CPUs restrict the enclave page cache to 128MB).
Though one can use system memory,
the resulting memory paging does not only produce
performance overhead but also introduces
more memory side-channels~\cite{sgxsidechannels}.

\paragraph{Sample-based data analysis}
Data sampling has many applications in data analysis 
from returning an approximate query result to
training a model using mini-batch stochastic gradient descent (SGD).
\ifFull
Sampling can be used for
approximating the results
when performing
the computation on the whole dataset is expensive
or not necessary, for example, when one is only interested in a
sample-based estimation of a statistic of a dataset
or an estimate answer to more complex queries
such as  frequent itemsets mining~\cite{Riondato:2014:EDA:2663597.2629586} and graph analysis~\cite{Riondato:2014:FAB:2556195.2556224}.\footnote{We note that we use sampling differently
from statistical approaches that treat the dataset $\cd$
as a sample from a population and use all records in $\cd$
to estimate parameters of the underlying population.}
The output of sample-based queries
is often associated with a confidence interval:
a probability bound on the error between
the true and the estimated result. Sample size
is often used as a parameter of the interval
with larger samples giving smaller chance of error.
(For example,
a bound on the probability that a result is either over or under represented in the sample
can be computed using a large deviation bound such as the Chernoff bound~\cite{DBLP:books/daglib/0012859}).
We consider various uses of sampling, including queries that
require a single sample, multiple samples such as bootstrapping statistics,
or large number of samples such as training of a neural network.
\else
Sampling can be used for
approximating results
when performing
the computation on the whole dataset is expensive
 (e.g., graph analysis or frequent itemsets~\cite{Riondato:2014:FAB:2556195.2556224,Riondato:2014:EDA:2663597.2629586})
\footnote{We note that we use sampling differently
from statistical approaches that treat the dataset $\cd$
as a sample from a population and use all records in $\cd$
to estimate parameters of the underlying population.}
or not needed
(e.g., audit of a financial institution by a regulator
based on a sample of the records).
We consider various uses of sampling, including queries that
require a single sample, multiple samples such as bootstrapping statistics,
or large number of samples such as training of a neural network.
\fi

Sampling-based queries provide:~\emph{Efficiency:} computing on a sample is
faster than on the whole dataset,
which fits the TEE setting, and
can be extended to process dataset samples in parallel with multiple TEEs.
\emph{Expressiveness:} a
large class of queries can be answered approximately using samples,
furthermore sampling (or mini-batching) is at the core of training modern machine learning models.
\emph{Privacy:} a query result from a sample
reveals information only about the sample and not the whole dataset.
Though intuitively privacy may come with sampling,
it is not always true. If a data scientist
knows indices of the records in the sample
used for a query, then given the query result
they learn more about records in that sample
than about other records.
However if sample indices are hidden then there is plausible deniability.
Luckily, differential privacy
takes advantage of privacy from
sampling and formally captures
it with \emph{privacy amplification}~\cite{Beimel2014,learnprivately,
Li:2012:SAD:2414456.2414474}.

\paragraph{Differential privacy}
Differential privacy (DP) is a rigorous definition of individual privacy
when a result of a query on the dataset is revealed.
Informally, it states that a single record does not significantly change the result
of the query.
Strong privacy can be guaranteed in
return for a drop in accuracy for simple statistical queries~\cite{privacybook}
and complex machine learning models~\cite{abadi,Bassily:2014:PER:2706700.2707412,DBLP:conf/iclr/McMahanRT018,pmlr-v37-wangg15,dpmp_dl}.
DP mechanisms come with a parameter $\epsilon$,
where higher $\eps$ signifies a higher privacy loss.

Amplification by sampling is a well known result in differential privacy.
Informally, it says that when an $\eps$-DP mechanism is applied
on a sample of size $\gamma n$
from a dataset $\cd$ of size $n$, $\sampleprob < 1$, then
the overall mechanism is $O(\gamma\eps)$-DP w.r.t.~$\cd$.
Small $\eps$ parameters reported from training of neural networks using DP SGD~\cite{abadi,DBLP:conf/iclr/McMahanRT018,dpmp_dl}
make extensive use of privacy amplification in their analysis.
Importantly, for this to hold they all require the sample identity to be hidden.

DP algorithms mentioned above are set in the trusted curator
model where hiding the sample is not a problem
as algorithm execution is not visible to an attacker (i.e., the data scientist who obtains the result
in our setting).
TEEs can be used only as an approximation of this model
due to the limitations listed above:
revealing memory access patterns of
a differentially-private algorithm can be enough to violate
or weaken its privacy guarantees.
Sampling-based DP algorithms fall in the second category
as they make an explicit assumption that
the identity of the sample is hidden~\cite{renyisampl,47557}.
If not, amplification based results cannot be applied.
If one desires
the same level of privacy, higher level of noise
will need to be added which would in turn reduce the utility of
the results. 

Differential privacy is attractive since it can keep track of the privacy loss over multiple queries. Hence,
reducing privacy loss of individual queries and supporting more queries as
a result, is an important requirement.
Sacrificing on privacy amplification by revealing sample identity is wasteful.

\paragraph{Data-oblivious sampling algorithms}
Query computation can be supported in a TEE
since samples are small compared to the dataset and can fit into private
memory of a TEE.
However, naive implementation of data sampling algorithms
is inefficient (due to random access to memory outside of TEE)
and insecure in our threat model (since sample
indices are trivially revealed).
Naively hiding sample identity would be to read a whole
dataset and only keep elements whose
indices happen to be in the sample.
This would require reading the entire dataset
for each sample (training of models usually requires small samples, e.g.,  0.01\% of the dataset).
This will also not be competitive in performance with shuffling-based approaches used today.

To this end, we propose novel algorithms
for producing data samples for two popular sampling approaches:
sampling without replacement and Poisson. Samples produced by shuffling-based sampling 
contain distinct elements, however elements may repeat between the samples.
Our algorithms are called data-oblivious~\cite{Goldreich:1987:TTS:28395.28416}
since the memory accesses they produce are independent
of the sampled indices.
Our algorithms are efficient as they require
only two data oblivious shuffles and one scan to produce $n/\samplesize$
samples of size $\samplesize$ that is sufficient for one epoch
of training.
An oblivious sampling algorithm would be used as follows:
$n/\samplesize$ samples are generated at once,
stored individually encrypted, and then loaded in a TEE on a per-query request.

\paragraph{Contributions}
\textbf{(i)} We propose a Private Sampling-based Query Framework
for querying sensitive data;
{\bf(ii)} We use differential privacy to show that sampling algorithms are an important
building block in privacy-preserving frameworks;
{\bf (iii)} We develop efficient and secure (data-oblivious) algorithms for two common sampling
techniques;
{\bf (iv)} We empirically show that for MNIST and CIFAR-10
using sampling algorithms for generating mini-batches during
differentially-private training achieves the same accuracy as shuffling,
even though sampling incurs smaller privacy loss than shuffling.

\section{Notation and Background}

A dataset $\cd$ contains $n$ elements; each element $e$ has a key and a value; keys are distinct in $[1,n]$.
If a dataset does not have keys, we use its element index in the array representation of $\cd$ as a key.

\ifFull
\else
\textbf{Trusted Execution Environment}\label{sec:sgx}
(TEE) provides strong protection guarantees
to data in its private memory: it is not visible
to an adversary who can control everything outside of the CPU,
e.g., even if it controls the operating system (OS) or the VM.
The private memory of TEEs (depending on the side-channel threat
model) is restricted to CPU registers (few kilobytes) or caches (32MB) or enclave page cache (128MB).
Since these sizes will be significantly smaller than usual datasets,
an algorithm is required to store the data
in the external memory.
Since external memory is controlled by an adversary (e.g., an OS), it can observe
its \emph{content} and the \emph{memory addresses} requested from a TEE.
Probabilistic encryption can be used to protect the \emph{content}
of data in external memory: an adversary seeing two ciphertexts
cannot tell if they are encryptions of the same element or 
a dummy of the same size as a real element.

Though the size of primary memory is not sufficient to process a dataset,
it can be leveraged for sample-based data analysis queries as follows.
When a query requires a sample,
it loads an encrypted sample from the
external memory into the TEE, decrypts it, 
performs a computation (for example, SGD),
discards the sample,
and either updates a local state (for example,
parameters of the ML model maintained in a TEE)
and proceeds to the next sample,
or encrypts the result of the computation under data scientist's secret key and
returns it.

\emph{Addresses (or memory access sequence)} requested by a TEE
can leak information about data.
Leaked information depends
on adversary's background knowledge (attacks based on
memory accesses have been shown for image and text processing~\cite{sgxsidechannels}).
In general, many (non-differentially-private and differentially-private~\cite{DBLP:journals/corr/abs-1807-00736})
algorithms leak their access pattern including sampling (see~\S\ref{sec:oblswo}).
\fi

\ifFull
\else
\textbf{Data-oblivious algorithms}
access memory in a manner that appears to be independent
of the sensitive data.
For example, sorting networks are data-oblivious
as compare-and-swap operators access the same array indices
independent of the array content, in contrast to quick sort.
Data-oblivious algorithms have been designed for
array access~\cite{Goldreich:1987:TTS:28395.28416,GoldreichO96,pathoram}, sorting~\cite{DBLP:conf/stoc/Goodrich14}, machine learning algorithms~\cite{Ohrimenko2016}
and several data structures~\cite{Wang:2014:ODS:2660267.2660314};
while this work is the first to consider sampling~algorithms.
The performance  goal of oblivious algorithms is to reduce the number of additional accesses to external
memory needed to hide real accesses.

Our sampling algorithms in \S\ref{sec:oblsampl} rely on an oblivious shuffle
$\oshuffle(\cd)$~\cite{melbshuffle}.
A shuffle rearranges elements
according to permutation $\pi$ s.t.
element at index $i$ is placed at location $\pi[i]$ after the shuffle.
An oblivious shuffle does the same except
the adversary observing its memory accesses
does not learn~$\pi$.
The Melbourne shuffle~\cite{melbshuffle}
makes $O(cn)$ accesses to external memory with private memory of size $O(\sqrt[c]{n})$.
This overhead is constant since non-oblivious shuffle need to make $n$ accesses.
Oblivious shuffle can use smaller private memory at the expense of more accesses
(see~\cite{DBLP:journals/corr/PatelPY17}).
It is important to note that while loading data into private memory,
the algorithm re-encrypts the elements
to avoid trivial comparison of elements before and after the shuffle.
\fi

\textbf{Differential privacy}
A randomized mechanism $\mech: D \rightarrow \mathcal{R}$ is $(\eps,\delta)$
differentially private~\cite{privacybook} if for any two neighbouring datasets
$\cd_0, \cd_1 \in D$  and for any subset of outputs $R \in \mathcal{R}$
it holds that $\Pr[\mech(\cd_0) \in \mathcal{R}] \le e^\eps \Pr[\mech(\cd_1) \in \mathcal{R}] + \delta$.
\ifFull
We consider 
substitute-one neighbouring relationship where $\cd_0$ and $\cd_1$ are of the same size
but are different in one element.
This relationship is natural for $\swo$ setting and data-oblivious setting
where an adversary knows the size
of the dataset.
In fact as we will see in \S\ref{sec:oblpoi}
it is hard to hide the size of Poisson sampling in
our setting and we choose to hide the number of samples instead.
\else
We use 
substitute-one neighbouring relationship where $|\cd_0| = |\cd_1|$
and $\cd_0, \cd_1$ are different in one element.
This relationship is natural for sampling without replacement and data-oblivious setting
where an adversary knows $|\cd|$.
As we see in \S\ref{sec:oblpoi}
hiding the size of Poisson sampling in
our setting is non-trivial and we choose to hide the number of samples instead.
\fi

Gaussian mechanism~\cite{privacybook} is a common way of
obtaining differentially private variant of real valued function $f: D \rightarrow \mathcal{R}$.
Let $\Delta_f$ be the $L_2$-sensitivity of $f$,
that is the maximum distance $\left\lVert  f(\cd_0) - f(\cd_1)\right\rVert_2$
between any $\cd_0$ and $\cd_1$.
Then, Gaussian noise mechanism is defined by $\mech(\cd) = f(\cd) + \mathcal{N}(0, \sigma^2)$
where $\mathcal{N}(0, \sigma^2\Delta_f^2)$ is a Gaussian distribution with mean 0 and standard deviation
$\sigma\Delta_f$. 
The resulting mechanism is $(\eps, \delta)$-DP
if $\sigma =  \sqrt{2\log (1.25/\delta)}/\eps$ for $\eps,\delta \in (0,1)$.

\ifFull
\paragraph{Query computation on a sample}
Sample computation with TEE proceeds as follows.
When a query requires a sample,
it loads an encrypted sample from the
external memory, decrypts it, 
performs a computation (for example, SGD),
discards the sample,
and either updates a local state (for example,
parameters of the ML model maintained in a TEE)
and proceeds to the next sample,
or encrypts the result of the computation under data scientist's key and
returns it.
\fi

\ifFull
\subsection{Sampling Algorithms}
\else
\textbf{Sampling methods}
\fi
Algorithms that operate on data samples
often require more than one sample.
For example, machine learning model training
proceeds in epochs where each epoch
processes multiple batches (or samples) of data.
The number of samples $\numsamples$ and sample size $\samplesize$ are usually chosen
such that $n \approx \numsamples \samplesize$ so that
every data element has a non-zero probability of being processed
during an epoch.
\ifFull

Algorithms that need data samples to operate
often require more than one sample.
For example, machine learning model training
proceeds in epochs where each epoch
processes multiple batches (or samples) of data.
The number of samples $\numsamples$ and sample size $\samplesize$ are usually chosen to be
such that $n \approx \numsamples \samplesize$ so that
every element of the data has a non-zero probability of being chosen
during an epoch.
To this end, we define
function $\samplepass_\alg(\cd, \sampleparam, \numsamples)$ that
produces samples $\sample_1, \sample_2, \ldots, \sample_\numsamples$.
We omit explicitly stating the randomness used in $\samplefunc_\alg$
but assume that every call uses a new seed.
\else
To this end, we define $\samplepass_\alg(\cd, \sampleparam, \numsamples)$ that
produces samples $\sample_1, \sample_2, \ldots, \sample_\numsamples$
using a sampling algorithm~$\alg$ and parameter $\sampleparam$,
where $\sample_i$ is a set of keys from $[1,n]$.
For simplicity we assume that $\samplesize$ divides $n$ and $\numsamples = n/\samplesize$.
We omit stating the randomness used in~$\samplepass_\alg$
but assume that every call uses a new seed.
We will now describe three sampling methods that
vary based
on element distribution within each sample and between the samples.
\fi
\ifFull

\paragraph{Sampling without replacement ($\swo$)}
Sampling without replacement produces a sample
by drawing $\samplesize$ distinct elements uniformly at random from a set $[1,n]$.
That is, the probability of a particular sample is $\frac{1}{n}\frac{1}{n-1}\ldots \frac{1}{n-\samplesize+1}$.
Let $\mathcal{F}_\swo^{n,\samplesize}$ be the
set of all $\swo$ samples of size $m$
from domain $[1,n]$.
Then $\samplefunc_\swo(\cd, \samplesize)$ returns a random
element of $\mathcal{F}_\swo^{n,\samplesize}$
and $\samplepass_\swo(\cd, \samplesize, \numsamples)$ draws $\numsamples$
samples
from $\mathcal{F}_\swo^{n,\samplesize}$
with replacement. That is, elements
cannot repeat within the same $\swo$ sample but can repeat between
the samples.
\else

{\emph{Sampling without replacement ($\swo$)}}
produces a sample
by drawing $\samplesize$ distinct elements uniformly at random from a set $[1,n]$,
hence probability of a sample $\sample$ is $\frac{1}{n} \frac{1}{n-1}\cdots \frac{1}{n-\samplesize+1}$.
Let $\mathcal{F}_\swo^{n,\samplesize}$ be the
set of all $\swo$ samples of size $m$
from domain $[1,n]$;
$\samplepass_\swo(\cd, \samplesize, \numsamples)$ draws $\numsamples$
samples
from $\mathcal{F}_\swo^{n,\samplesize}$
with replacement: elements
cannot repeat within the same sample but can repeat between
the samples.
\fi
\ifFull

\paragraph{$\poisson$ Sampling}
A Poisson sample $\sample$ is constructed
by independently adding each element from the set $[1,n]$ with probability $\sampleprob$,
that is $\Pr(j \in \sample) = \sampleprob$, $\forall j \in [1,n]$.
Hence, probability of a sample $\sample$ is $\Pr_\sampleprob(\sample)= \sampleprob^{|s|}(1-\sampleprob)^{n-|s|}$.
Let $\mathcal{F}_\poisson^{n,\sampleprob}$ be the
set of all Poisson samples from domain $[1,n]$.
Then $\samplefunc_\poisson(\cd, \sampleprob)$ returns an element
$s$ from $\mathcal{F}_\poisson^{n,\samplesize}$ supported by probability
distribution
$\Pr_\sampleprob(\sample)$
and $\samplepass_\poisson(\cd, \sampleprob, \numsamples)$ draws $\numsamples$
elements with replacement from  $\mathcal{F}_\poisson^{n,\sampleprob}$
also with probability $\Pr_\sampleprob$.

Note that the size of a Poisson sample is a random variable
and is $\sampleprob n$ is an average size of a sample.
Similar to $\swo$ each Poisson sample contains distinct elements,
while elements can repeat between Poisson samples.
\else
\emph{Poisson Sampling ($\poisson$)}
$\sample$ is constructed
by independently adding each element from $[1,n]$ with probability $\sampleprob$,
that is $\Pr(j \in \sample) = \sampleprob$.
Hence, probability of a sample $\sample$ is $\Pr_\sampleprob(\sample)= \sampleprob^{|s|}(1-\sampleprob)^{n-|s|}$.
Let $\mathcal{F}_\poisson^{n,\sampleprob}$ be the
set of all Poisson samples from domain $[1,n]$.
Then, $\samplepass_\poisson(\cd, \sampleprob, \numsamples)$ draws~$\numsamples$
elements with replacement from $\mathcal{F}_\poisson^{n,\sampleprob}$.
The size of a Poisson sample is a random variable
and $\sampleprob n$ on average.
\fi
\ifFull

\paragraph{Sampling via Shuffling ($\shuffle$)}
Shuffling is a common approach for obtaining batches for mini-batch
training as it is simple to implement and, once shuffled,
can support parallel and sequential access to the batches.
It proceeds as follows: shuffle $\cd$, split the shuffled data into
batches of size $\samplesize$, and treat each batch as a sample, producing $\numsamples= n/\samplesize$
samples in total.
If more than $\numsamples$ samples are required, the procedure is repeated, producing the next $\numsamples$ samples.

A routine $\samplepass_\shuffle(\cd, \samplesize)$ returns
samples $\sample_1, \sample_2, \ldots, \sample_\numsamples$, each of size $\samplesize$ such that $\cup_{i_\in[1..\numsamples]} \sample_i = \cd$.
Function
$\samplefunc_\poisson(\cd, \sampleprob)$ is defined
using $\samplepass_\shuffle(\cd, \samplesize)$
where samples are saved and returned one at a time. After $\numsamples$ calls,
$\samplepass_\shuffle(\cd, \samplesize)$ is called again.
Similar to $\swo$ or $\poisson$, each sample
contains distinct elements, however in contrast to $\swo$ or $\poisson$,
a sequence of $\numsamples$ samples (i.e., samples from from a single shuffle)
always contain distinct elements between the samples.
\else
\emph{Sampling via $\shuffle$}
is common for obtaining mini-batches for SGD in practice.
It shuffles $\cd$ and splits it in
batches of size $\samplesize$.
If more than $\numsamples$ samples are required, the procedure is repeated.
Similar to $\swo$ or $\poisson$, each sample
contains distinct elements, however in contrast to them,
a sequence of $\numsamples$ samples contain distinct elements between the samples.
\fi

\section{Privacy via Sampling and Differential privacy}

\label{sec:dpsampling}
Privacy amplification of differential privacy captures
the relationship of performing
analysis over a sample vs.~whole dataset.
Let $\mech$ be a randomized mechanism that
is $(\epsilon,\delta)$-DP
and let $\samplefunc$ be a random sample
from dataset $\cd$ of size $\sampleprob n$,  where {$\sampleprob < 1$}
is a sampling parameter.
Let $\mech' = \mech \circ \samplefunc$ be a mechanism that
applies $\mech$ on a sample of $\cd$.
Then, informally, $\mech'$ is $(O(\sampleprob \epsilon), \gamma\delta)$-DP~\cite{Beimel2014,
Li:2012:SAD:2414456.2414474}.

\paragraph{Sampling} For Poisson and sampling without replacement
$\eps'$ of $\mech'$
is  $\log(1 + \sampleprob (e^\eps -1))$~\cite{Li:2012:SAD:2414456.2414474}
and $\log(1 + \samplesize/n(e^\eps - 1))$~\cite{balle_subsampling}, respectively.
We refer the reader to
Balle~\emph{et al.}~\cite{balle_subsampling}
who provide a unified framework for
studying amplification of these
sampling mechanisms.
Crucially all amplification results assume that \emph{the sample is hidden during the analysis
as otherwise amplification results cannot hold}. That is, if the
keys of the elements of a
sample are revealed, $\mech'$ has the same~$(\eps,\delta)$ as~$\mech$.

Privacy loss of executing a sequence of~DP mechanisms
can be analyzed using several
approaches.
Strong composition theorem~\cite{privacybook} states that running $T$
$(\eps, \delta)$-mechanisms 
would be
$(\eps\sqrt{2T \log(1/\delta'')} + T \epsilon (e^\eps - 1),T\delta + \delta'')$-DP,
$\delta''\ge 0$.
Better bounds can be obtained
if one takes advantage of the underlying DP mechanism.
Abadi~\emph{et al.}~\cite{abadi}
introduce the moments accountant
that leverages the fact that $\mech'$
uses Poisson sampling and applies Gaussian noise
to the output.
They obtain $\eps'$ = $O(\sampleprob \eps \sqrt{T})$,~$\delta' = \delta$.

\paragraph{Shuffling}
Analysis of differential private parameters
of $\mech'$
that operates on samples
obtained from shuffling is different.
Parallel composition by McSherry~\cite{McSherry:2009:PIQ:1559845.1559850}
can be seen as the privacy ``amplification'' result
for shuffling.
It states that running
$T$ algorithms in parallel
on \emph{disjoint} samples of the dataset
has $\epsilon' = \max_{i\in[1,T]} \epsilon_i$
where $\epsilon_i$ is the parameter of the $i$th mechanism.
It is a significantly better result
than what one would expect from
using DP composition theorem,
since it relies on the fact that samples are disjoint.
If one requires multiple passes over a dataset
(as is the case with multi-epoch training),
strong composition theorem can be used
with parallel composition.

\paragraph{Sampling vs.~Shuffling DP Guarantees}
We bring the above results together in Table~\ref{tbl:epstheor} to compare the parameters
of several sampling approaches.
As we can see, sampling based approaches for general DP mechanisms
give an order of $O(\sqrt{\samplesize/n})$ smaller
epsilon than shuffling based approaches.
It is important to note that sampling-based approaches
assume that the indices (or keys) of the dataset elements used
by the mechanism remain secret.
In~\S\ref{sec:oblsampl} we develop algorithms with this property.

\begin{table}[t]
\small
\begin{center}
\caption{Parameters $(\eps',\delta')$ of mechanisms that
use $(\eps,\delta)$-DP mechanism $\mech$ with one of the three sampling techniques
with a sample of size $\samplesize$
from a dataset of size $n$
and $\sampleprob = \samplesize/n$
for Poisson sampling, where $\eps' <1$, $\delta'' > 0$,
$T$ is the number of samples in an epoch,
$E$ is the number of epochs.}
\begin{tabular}{ r || c | c }
\multirow{2}{*}{\textbf{Sampling mechanism}}& \multicolumn{2}{c}{\bf \# analyzed samples of size $\samplesize$} \\
\cline{2-3}
& $T \le n/\samplesize$ & $ T = E n/\samplesize$, $E \ge1$ \\
\hline
\hline
Shuffling & $\eps,\delta$ & $O(\eps\sqrt{E \log(1/\delta'')}),E\delta + \delta'')$   \\
\cline{2-3}
Poisson, SWO &  \multicolumn{2}{c}{$O(\eps\sampleprob\sqrt{T \log(1/\delta'')}),T\sampleprob\delta + \delta'')$}  \\
\hline
Poisson \& Gaussian distribution~\cite{abadi} & \multicolumn{2}{c}{$O(\gamma \eps \sqrt{T}),\delta$}  \\
\end{tabular}
\label{tbl:epstheor}
\end{center}
\end{table}

\paragraph{Differentially private SGD}
We now turn our attention to a differentially private mechanism for
mini-batch stochastic gradient descent computation.
The mechanism is called NoisySGD~\cite{Bassily:2014:PER:2706700.2707412, 6736861}
and when applied instead of non-private mini-batch SGD
allows for a release of a machine learning model with
differential privacy guarantees on the training data.
For example, it has been applied
in Bayesian learning~\cite{pmlr-v37-wangg15}
and
to train deep
learning~\cite{abadi, DBLP:conf/iclr/McMahanRT018, dpmp_dl}
and
logistic regression~\cite{6736861} models.

\label{sec:noisysgd}

It proceeds as follows.
Given a mini-batch (or sample) the gradient of every element
in a batch
is computed and the L2 norm of the gradient is clipped
according to a clipping parameter $C$.
Then a noise is added to the sum
of the (clipped) gradients of all the elements
and the result is averaged over
the sample size.
The noise added to the result is from Gaussian
distribution parametrized with~$C$
and a noise scale parameter $\sigma$: $\mathcal{N}(0, \sigma^2 C^2)$.
The noise is proportional to the
sensitivity of the sum of gradients
to the value of each element in the sample.
The amount of privacy budget that a single batch
processing,
also called subsampled Gaussian mechanism, incurs depends on the parameters of the noise distribution
and how the batch is sampled.
The model parameters are iteratively updated after every NoisySGD processing.
The number of iterations and
the composition mechanism used to keep track
of the privacy loss determine
the DP parameters of the overall training process.

Abadi~\emph{et al.}~\cite{abadi} report analytical results assuming Poisson sampling
but use
shuffling to obtain the samples in the evaluation.
Yu~\emph{et al.}~\cite{dpmp_dl} point out the discrepancy between analysis
and experimental results in~\cite{abadi}, that is, the
reported privacy loss is underestimated due to the use of shuffling.
Yu~\emph{et al.} proceed to analyze shuffling and sampling but also use shuffling
in their experiments.
Hence, though analytically Poisson and SWO sampling
provide better privacy parameters
than shuffling, there is no evidence
that the accuracy is the same between the approaches in practice.
We fill in this gap in \S\ref{sec:evaluation} and show that
for the benchmarks we have tried it is indeed the case.

\section{Oblivious Sampling Algorithms}
\label{sec:oblsampl}
In this section, we develop data-oblivious algorithms for generating
a sequence of samples from a dataset~$\cd$ such that the total number of
samples is sufficient for a single epoch of a training algorithm.
Moreover, our algorithms will access the original
dataset at indices that appear to be independent
of how elements are distributed across the samples.
As a result, anyone observing their memory accesses
cannot
identify, how many and which samples
each element of $\cd$ appears in.

\ifFull
In this section we develop new oblivious sampling algorithms
for sampling w/o replacement and Poisson sampling.  
\fi

\subsection{Oblivious sampling without replacement ($\swo$)}
\label{sec:oblswo}

We introduce a definition of an \emph{oblivious sampling algorithm}:
oblivious $\samplepass_\swo(\cd, \samplesize)$ is a randomized
algorithm that returns $\numsamples = n/\samplesize$ $\swo$
samples from $\cd$ and produces memory accesses
that are indistinguishable
between invocations for all datasets of size $n = |\cd|$
and generated samples.

As a warm-up, consider the following naive way of generating a single $\swo$ sample
of size $\samplesize$
from dataset $\cd$ stored in external memory
of a TEE:
generate~$\samplesize$ distinct random keys from $[1,n]$ and load
from external memory elements of $\cd$ that are stored at those indices.
This trivially reveals the sample to an observer of memory accesses.
A secure but inefficient way would be to load
$\cd[l]$ for $\forall l \in [1,n]$ and, if $l$ matches one of the $\samplesize$
random keys, keep $\cd[l]$ in private memory. This incurs $n$
accesses to generate a sample of size~$\samplesize$.
Though our algorithm will also make a linear number of
accesses to $\cd$, it will amortize this cost
by producing $n/\samplesize$ samples.

The high level description of our secure and efficient 
algorithm for producing $\numsamples$ samples is as follows.
Choose~$\numsamples$ samples from $\mathcal{F}_\swo^{n,\samplesize}$,
numbering each sample with an identifier $1$ to $\numsamples$;
the keys within the samples (up to a mapping) will represent the keys of
elements used in the samples of the output. 
Then, while scanning $\cd$, replicate elements depending on how many samples
they should appear in and associate each replica with its sample id.
Finally, group elements according to sample~ids.

\paragraph*{Preliminaries}
Our algorithm relies on a primitive that can
efficiently draw $\numsamples$ samples from $\mathcal{F}_\swo^{n,\samplesize}$
(denoted via~$\swo.\init(n,\samplesize)$). It also provides a function $\swo.\samplemember(i,j)$
that returns $\True$ if key $j$ is in the $i$th sample and $\False$ otherwise.
This primitive can be instantiated using~$\numsamples$ pseudo-random
permutations~$\rho_i$ over $[1,n]$. Then sample $i$ is defined
by the first $\samplesize$ indices of the permutation, i.e., element
with key~$j$
is in the sample $i$ if $\rho_i(j) \le \samplesize$.
This procedure is described in more detail in Appendix~\S\ref{app:swo}.

We will use $r_j$ to denote the number of samples 
where key $j$ appears in, that is $r_j = |\{i~|~ \samplemember(i,j), \forall i \in [1, \numsamples], \forall j \in [1, n]\}|$.
It is important to note that samples drawn above are used as a template for a valid SWO sampling
(i.e., to preserve replication of elements across the samples).
However, the final samples $\sample_1, \sample_2, \ldots, \sample_\numsamples$ returned by the algorithm
will be instantiated with keys that are determined using
function $\pi'$ (which will be defined later).
In particular, for all samples, if $\samplemember(i,j)$ is true then $\pi'(j) \in \sample_i$.

\paragraph*{Description}
The pseudo-code in Algorithm~\ref{alg:oblswo} provides the details of the method.
It starts with dataset $\cd$ obliviously shuffled according to a random secret permutation $\pi$ (Line~\ref{line:swoshuf1}).
Hence, element~$e$ is stored (re-encrypted) in $\cd$ at index $\pi(e.\key)$.
The next phase replicates elements
such that for every index $j\in[1,n]$ there is an element (not necessarily with key $j$) that is replicated $r_j$
times (Lines~\ref{line:swobegrep}-\ref{line:swoendrep}).
The algorithm maintains a counter~$l$ which keeps the current index of the scan in the array
and $\nexte$ which stores the element read from $l$th index.

\begin{wrapfigure}{tr}{0.5\textwidth}
\vspace{-10pt}
\hspace{5pt}
\begin{minipage}[tr]{0.45\textwidth}
\begin{algorithm}[H]
	\caption{Oblivious $\samplepass_\swo(\cd, \samplesize)$: takes an encrypted dataset $\cd$ and returns $\numsamples=n/\samplesize$ SWO samples of size $\samplesize$, $n=|\cd|$.}
		\label{alg:oblswo}
		\begin{algorithmic}[1]
\STATE $\cd \leftarrow \oshuffle(\cd)$ \label{line:swoshuf1}
\STATE $\swo.\init(n,\samplesize)$ \label{line:swosamples}
\STATE $\mathcal{S} \leftarrow []$, $j \leftarrow 1$, $l \leftarrow 1$, $e \leftarrow \cd[1]$, $\nexte \leftarrow \cd[1]$
\WHILE {$l \le n$} \label{line:swobegrep}
	\FOR{$i \in [1,\numsamples]$}
		\IF{$\swo.\samplemember(i, j)$}\label{line:swoifsample}
			\STATE $\mathcal{S}.\append(\reenc(e),\enc(i))$
			\STATE $l \leftarrow l + 1$
			\STATE $\nexte \leftarrow \cd[l]$
		\ENDIF
	\ENDFOR
	\STATE $e \leftarrow \nexte$
	\STATE $j \leftarrow j + 1$
\ENDWHILE \label{line:swoendrep}
\STATE $S \leftarrow \oshuffle(\mathcal{S})$ 
\STATE $\forall i \in [1,\numsamples]: \sample_i \leftarrow  []$ \label{line:swobeggroup}
\FOR{$p \in S$}
	\STATE $(c_e, c_i) \leftarrow p$, $i \leftarrow \dec(c_i)$
	\STATE $\sample_i \leftarrow \sample_i.\append(c_e)$
\ENDFOR \label{line:swoendgroup}
\STATE Return $\sample_1, \sample_2, \ldots, \sample_\numsamples$
\end{algorithmic}
\end{algorithm}
\end{minipage}
\end{wrapfigure}

Additionally the algorithm maintains element $e$ which is an element that currently is being replicated.
It is updated to $\nexte$
as soon as sufficient number of replicas is reached. The number of times $e$ is replicated depends on the number
of samples element with key $j$ appears in. 
Counter $j$ starts at $1$ and is incremented after element $e$ is replicated
$r_j$ times. At any given time, counter $j$ is an indicator of the number of distinct elements
written out so far. Hence, $j$ can reach $n$ only if every element appears in exactly one sample. On the other hand,
the smallest $j$ can be is $\samplesize$, this happens when all $\numsamples$ samples are identical.

Given the above state, the algorithm reads an element into $\nexte$,
loops internally through $i \in [1..\numsamples]$:
if current key $j$ is in $i$th sample
it writes out an encrypted tuple $(e,i)$
and reads the next element from~$\cd$ into $\nexte$.
Note that $e$ is re-encrypted every time it is written out in order to hide
which one of the elements read so far is being written out.
After the scan, the tuples are obliviously shuffled.
At this point, the sample id $i$ of each tuple is decrypted and
used to (non-obliviously) group elements that belong to the same sample
together, creating the sample output $\sample_1..\sample_\numsamples$~(Lines \ref{line:swobeggroup}-\ref{line:swoendgroup}).

\newcommand{\mapping}{\mathsf{m}}

We are left to derive the mapping $\mapping$ between keys
used in samples drawn in Line~\ref{line:swosamples} and elements
returned in samples $\sample_1..\sample_\numsamples$.
We note that $\mapping$ is not explicitly used during the algorithm
and is used only in the analysis.
From the algorithm we see that
$\mapping(l) = \pi^{-1}(1 + \sum_{j=1}^{l-1} r_j)$,
that is $\mapping$ is derived from $\pi$ with shifts due to replications of preceding keys.
(Observe that if
every element appears only in one sample $\mapping(l) = \pi^{-1}(l)$.)
We show that $\mapping$ is injective and random (Lemma~\ref{lemma:injectmap})
and, hence, $\sample_1..\sample_\numsamples$
are valid $\swo$ samples.

\paragraph{Example}
Let $\cd = \{(1,A)$, $(2,B),(3,C)$, $(4,D),(5,E),(6,F)\}$, where
$(4,D)$ denotes element $D$ at index $4$ (used also as a key),
$\samplesize = 2$, and
randomly drawn samples in $\swo.\init$ are $\{1,4\}$, $\{1,2\}$, $\{1,5\}$.
Suppose $\cd$ after the shuffle is $\{(4,D),$ $(1,A),$ $(5,E),$ $(3,C),$ $(6,F),$ $(2,B)\}$.
Then, after the replication $\mathcal{S}  = \{((4,D),1),$ $((4,D),2),$ $((4,D),3),$ $((3,C),2),$ $((6,F),1),$ $((2,B),3)\}$
where the first tuple $((4,D),1)$ indicates that $(4,D)$ appears in the first sample.

\paragraph{Correctness}
We show that samples returned by the algorithm
correspond to samples drawn randomly from $\mathcal{F}_\swo^{m,n}$.
We argue that samples returned by the oblivious $\samplepass_\swo$
are identical to those drawn truly at random from $\mathcal{F}_\swo^{m,n}$ up
to key mapping $\mapping$
and then show that $\mapping$ is injective and random in Appendix~\ref{app:swo}.
For every key $j$ present in the drawn samples
there is an element with key $\mapping(j)$ that
is replicated $r_j$ times and is associated with
the sample ids of $j$. Hence, returned samples, after being grouped,
are exactly the drawn samples where every key $j$
is substituted with an element with key~$\mapping(j)$.

\ifFull
\begin{lemma}
\label{lemma:injectmap}
Let $\pi$ be a permutation over $n$ elements,
$\forall j \in[1,n], r_j \in [0,\numsamples]$ such that $\sum^n r_j = n$
and $\mathcal{D} = \{j~|~r_j \ge 1\}$.
Let $\mapping(j) = \pi^{-1}(1 + \sum_{l=1}^{l-1} r_j)$.
Then $\mapping$ evaluated on keys in $\mathcal{D}$ is an injective random function
over $[1,n]$.
\end{lemma}
\begin{proof}
The statement follows from two observations: $\pi^{-1}$ is a permutation
and $\pi^{-1}$ is evaluated only on distinct elements from a set $[1,n]$.

The second observation is true since function from $l$ to $1 + \sum_{j=1}^{l-1} r_j$,
when evaluated on $l\in \mathcal{D}$, is injective
as it is strictly monotonic.
Moreover, $(1 + \sum_{j=1}^{l-1} r_j)\le n$
since $\sum^n r_j =  n$.

Co-domain of $\mapping$
appears independent of its input since
it is a subset of
the output of a random permutation function $\pi$
that has these properties by definition.
\end{proof}
\fi

\ifFull
\textbf{Security}
The adversary observes an oblivious shuffle, a scan where an element is read and an encrypted pair is written,
another oblivious shuffle and then a scan that reveals the sample identifiers. All patterns except for
revealing of the sample identifiers are
independent of the data.
We are left to argue that revealing sample ids and their locations does not reveal information
about the data nor the samples.
First note that there are $\samplesize$ copies of sample ids $1,2,\ldots,\numsamples$ associated with a ciphertext,
hence data-independent.
Second, note that locations of the revealed identifiers are random according
to the permutation chosen in the second shuffle step.
Since the permutation of the shuffles are hidden, the adversary does not learn the location of the tuple before and after the shuffle.
\else
%\vspace{-2pt}
\paragraph{Security and performance}
The adversary observes an oblivious shuffle, a scan where an element is read and an encrypted pair is written,
another oblivious shuffle and then a scan that reveals the sample identifiers. All patterns except for
revealing of the sample identifiers are
independent of the data and sampled keys. We argue security further in~\S\ref{app:swo}.
{Performance} of oblivious $\swo$ sampling is dominated
by two oblivious shuffles and the non-oblivious grouping,
replication scan has linear cost.
Hence, our algorithm produces $\numsamples$ samples in time $O(cn)$
with private memory of size $O(\sqrt[c]{n})$.
Since a non-oblivious version would require $n$ accesses,
our algorithm has a constant overhead for small $c$.
\fi

\ifFull
\textbf{Performance}
Performance of oblivious $\swo$ sampling is dominated
by two oblivious shuffles and the non-oblivious sorting (or grouping) in Lines~\ref{line:swobeggroup}-\ref{line:swoendgroup}
since the replication scan has linear cost.
Hence, our algorithm produces $\numsamples$ in time $O(cn)$
with private memory of size $O(\sqrt[c]{n})$.
Since a non-oblivious version would require $n$ accesses,
our algorithm has a constant overhead for small $c$.
Importantly, neither oblivious shuffle nor sample algorithms
perform random accesses to outsourced
memory. That is, all memory accesses
they produce are deterministic.
\fi

\paragraph{Observations}
We note that if more than $\numsamples$ samples of size $\samplesize = n/\numsamples$
need to be produced, one can invoke the algorithm multiple times using different randomness.
Furthermore, Algorithm~\ref{alg:oblswo} can produce samples
of varying sizes $\samplesize_1, \samplesize_2,.., \samplesize_\numsamples$ ($n = \sum{\samplesize_i}$)
given as an input.
The algorithm itself will remain the same. However,
in order to determine
if~$j$~is in sample $i$ or not, $\mathsf{samplemember}(i,j)$
will check if~$\rho_i(j) \le m_i$ instead of $\rho_i(j) \le m$.

\subsection{Oblivious Poisson sampling}
\label{sec:oblpoi}

Performing Poisson sampling obliviously requires not only hiding
access pattern but also the size of the samples.
Since in the worst case
the sample can be of size $n$, each sample will need to be
padded to $n$ with dummy elements. Unfortunately generating $\numsamples$ samples each padded to size
$n$ is impractical.
Though samples of size $n$ are unlikely, revealing some upper bound on sample
size would affect the security of the algorithms relying on Poisson sampling.

Instead of padding to the worst case, we choose to hide the number of samples
that are contained within an $n$-sized block of data (e.g., an epoch).
In particular, our oblivious Poisson sampling returns~$S$ that consists
of samples $\sample_1, \sample_2, \ldots, \sample_{\numsamples'}$
where $\numsamples'\le \numsamples$ such that  $\sum_{i\in[1,\numsamples']} |\sample_i| \le n$.
The security of sampling relies on hiding $\numsamples'$
and the boundary between the samples,
as otherwise an adversary can estimate sample sizes.

The algorithm, presented in Algorithm~\ref{alg:oblpoisson}, proceeds similar to $\swo$
except every element, in addition to being associated with
a sample id, also stores its position
in final $S$. The element and the sample id are kept private
while the position is used to order the elements.
It is then up to the queries that operate
on the samples inside of a TEE (e.g., SGD computation) to
use sample id while scanning~$S$ to determine
the sample boundaries.
The use of $\samplepass_\poisson$ by the queries
has to be done carefully without revealing
when the sample is actually used as this would reveal the boundary
(e.g., while reading the elements during an epoch, one needs
to hide after which element the model is updated).

\ifFull
We will again assume that
random samples, according to Poisson in this case,
can be drawn efficiently
and describe how in the next section.
\else
We assume that
that samples from $\mathcal{F}_\poisson^{n,\sampleprob}$
can be drawn efficiently
and describe how in Appendix~\S\ref{app:poi}.
\fi
The algorithm relies on two functions that have
access to the samples:
$\getsamplesize(i)$ and $\getsamplepos(i, l)$
which return the size of the $i$th sample and the position
of element $l$ in $i$th sample.
The algorithm uses the former to compute $\numsamples'$ and creates replicas for samples with identifiers from
1 to $\numsamples'$. 
The other changes to the Algorithm~\ref{alg:oblswo} are
that $\mathcal{S}.\append(\enc(e),\enc(i))$ is substituted
with  $\mathcal{S}.\append(\enc(e),\enc(i), \enc(\pos))$ 
where $\pos = \sum_{i' <i} {\getsamplesize(i')} + \getsamplepos(i, l)$.
If the total number of elements in the first $\numsamples'$ samples is less than $n$,
the algorithm appends dummy elements to $\mathcal{S}$.
$\mathcal{S}$ is then shuffled. After that positions $\pos$ can be decrypted and
sorted (non-obliviously) to bring elements from the same samples together.
In a decrypted form this corresponds to samples
ordered one after another sequentially, following with $\dummy$ elements if applicable.

\begin{algorithm}[t]
\begin{algorithmic}[1]
	\caption[]{Oblivious $\samplepass_\poisson(\cd, \sampleprob)$: takes an encrypted dataset $\cd$ and returns Poisson sample(s) with parameter $\sampleprob$, $n=|\cd|$}
	\label{alg:oblpoisson}
\STATE $\cd \leftarrow \oshuffle(\cd)$
\STATE $\poisson.\init(n,\sampleprob)$
\STATE $\mathcal{S} \leftarrow []$
\STATE $j \leftarrow 1$, $l \leftarrow 1$, $e \leftarrow \cd[1]$, $\nexte \leftarrow \cd[1]$
\STATE $\numsamples' \leftarrow 1$, ${\cursize}\leftarrow \poisson.\getsamplesize(1) $
\WHILE {$ \cursize+ \poisson.\getsamplesize(\numsamples' + 1) \le n$ and $\numsamples' + 1 \le \numsamples$}
  \STATE $\numsamples' \leftarrow \numsamples' + 1$
  \STATE $\cursize \leftarrow \cursize +  \poisson.\getsamplesize(\numsamples')$
\ENDWHILE
\WHILE {$j \le \cursize$}
	\FOR{$i \in [1,\numsamples']$}
		\IF{$\poisson.\samplemember(i, l)$}
			\STATE $\pos \leftarrow \sum_{i' < i} \poisson.\getsamplesize(i') + \poisson.\getsamplepos(i,l)$
			\STATE $\mathcal{S}.\append(\reenc(e),\enc(i), \enc(\pos))$
			\STATE $j \leftarrow j + 1$
			\STATE $\nexte \leftarrow \cd[j]$
		\ENDIF
	\ENDFOR
	\STATE $e \leftarrow \nexte$
	\STATE $l \leftarrow l + 1$
\ENDWHILE
\FOR {$j \in [\cursize+1, n]$}
	\STATE $\mathcal{S}.\append(\enc(\dummy),\enc(0), \enc(j))$
\ENDFOR
\STATE $S \leftarrow \oshuffle(\mathcal{S})$
\STATE Decrypt $\pos$ (last part) of every tuple in ${S}$ and use it to sort the encrypted elements\label{line:poisort}\label{line:posreveal}
\STATE Return $S$
\end{algorithmic}
\end{algorithm}

\ifFull
\subsection{Sampling primitives}
\label{sec:sampleprimitives}
Sampling with replacement from domain $[1,n]$ can be instantiated by choosing a random
permutation~$\rho$ over the same domain. Then, the sample is defined
by elements that are mapped to the first $\samplesize$ elements, i.e., element $j$
is in the sample if $\rho(j) \le \samplesize$.
This procedure is described in Figure~\ref{alg:samplmemswo} for $\numsamples$ samples.
During the $\init$ call, random permutations are chosen (e.g.,
in the real implementation this would correspond to choosing
a random seed and then deriving $\numsamples$ seeds
for each permutation).
The second function $\samplemember(i,j)$ returns $\True$
or $\False$ depending on whether
$j$ is in the $i$th sample or not.
\begin{algorithm}[t]
\begin{algorithmic}
	\caption[]{Instantiation of $\swo$ sampling for $\numsamples= n/\samplesize$ samples drawn from $\mathcal{F}_\swo^{n,\samplesize}$}
	\label{alg:samplmemswo}
\STATE $\init(n, \samplesize)$:
choose random permutations with domain $[1,n]$: $\rho, \rho_2, \ldots, \rho_\numsamples$
\STATE $\samplemember(i,j)$: If $\rho_i(j) \le \samplesize$ return $\True$, else $\False$
\end{algorithmic}
\end{algorithm}
\fi

\ifFull
\instantpoisson
\fi

\newcommand{\samplmempalg}{
\begin{algorithm}[h]
\begin{algorithmic}
	\caption[]{Instantiation of $\poisson$ sampling for $\numsamples = n\sampleprob$ samples drawn from $\mathcal{F}_\poisson^{n,\sampleprob}$}
	\label{alg:samplmemp}
\STATE $\init(n,\sampleprob)$:
\STATE ~~ choose random permutations with domain $[1,n]$: $\rho_1, \rho_2, \ldots, \rho_\numsamples$
\STATE ~~  $\forall i \in[1,\numsamples]$, $M_i \leftarrow \Binom(n,\sampleprob)$
\STATE $\samplemember(i,j)$: If $\rho_i(j) \le M_i$ return $\True$, else $\False$
\STATE$\getsamplesize(i)$: return $M_i$
\STATE $\getsamplepos(i, l)$: return $\rho_i(l)$
\end{algorithmic}
\end{algorithm}
}
\ifFull
\samplmempalg
\fi

\section{Experimental results}
\label{sec:evaluation}
The goal of our evaluation is to understand the impact of sampling
on the accuracy of training of neural network models
and their differentially private variants,
\ifFull
However, sampling with replacement or Poisson
sampling incurs smaller epsilon
and hence could be a preferred option in practice
if the performance of sampling mechanisms
is on par with shuffling.
In this section we show that this is the case.
\fi
We show that accuracy of all sampling mechanisms is
the same while shuffling has the highest privacy loss.

%\paragraph{Implementation}
We use TensorFlow v$1.13$
and TensorFlow Privacy library~\cite{dptflink}
for DP training.
We implement non-oblivious
$\swo$ and $\poisson$ sampling mechanisms
since accuracy of the training procedure
is independent of sampling implementation.
We report an average of 5 runs for each experiment.

Our implementation relies on DP optimizer from~\cite{dptflink}
which builds on ideas from~\cite{abadi}
to implement noisySGD as described in \S\ref{sec:noisysgd}.
Note that this procedure is independent of the
sampling mechanism behind how the batch
is obtained.
The only exception is $\poisson$ where
the average is computed using a fixed sample size
($\sampleprob \times n$) vs.~its
real size as for the other two sampling mechanisms.
We set the clipping parameter to 4, $\sigma = 6$, $\delta = 10^{-5}$.
For each sampling mechanism
we use a different privacy accountant
to compute exact total $\eps$ as opposed to
asymptotical guarantees in Table~\ref{tbl:epstheor}.
For shuffling we use~\cite{dpmp_dl,McSherry:2009:PIQ:1559845.1559850};
for Poisson sampling~\cite{dptflink};
and for $\swo$ we implement the approach from~\cite{renyisampl}.
\ifFull
We note that \cite{abadi,dpmp_dl} and \cite{renyisampl}
already show that specialized privacy accountants
are better than strong composition theorem~\cite{privacybook},
hence,we omit
the use of strong composition in our evaluation.
Observe that training itself does not depend on the privacy
accountant. It may depend
only if one uses privacy loss as a stopping criteria
for the training as opposed to the number of epochs.
\fi

\textbf{MNIST}
dataset contains 60,000 train and 20,000 test images of ten digits with
the classification tasks of determining which digit an image corresponds to.
We use the same model architecture as~\cite{abadi} and~\cite{dpmp_dl}.
It is a feed-forward neural network comprising of a single hidden layer with 1000 ReLU units and the output layer is softmax of 10 classes corresponding to the 10 digits. The loss function computes cross-entropy loss.
During training we sample data using shuffling,
sampling without replacement and Poisson.
For the first two we use batch size $\samplesize =600$, $\sampleprob = 0.01$
and $\samplesize =200$, $\sampleprob = 0.003$ in Figure~\ref{fig:eps}.
Each network is trained for 100 epochs.
We report the results in Table~\ref{tbl:acc} (left).
We observe that sampling mechanism does not change accuracy for this benchmark.

\ifFull
For example the test accuracy for MNIST for $\shuffle$,
$\poisson$, and $\swo$ is $97.5\%$, $97.47\%$ and $97.43\%$ respectively.
\fi

\begin{wrapfigure}{r}{0.4\textwidth}
	\centering
	\includegraphics[scale=0.5]{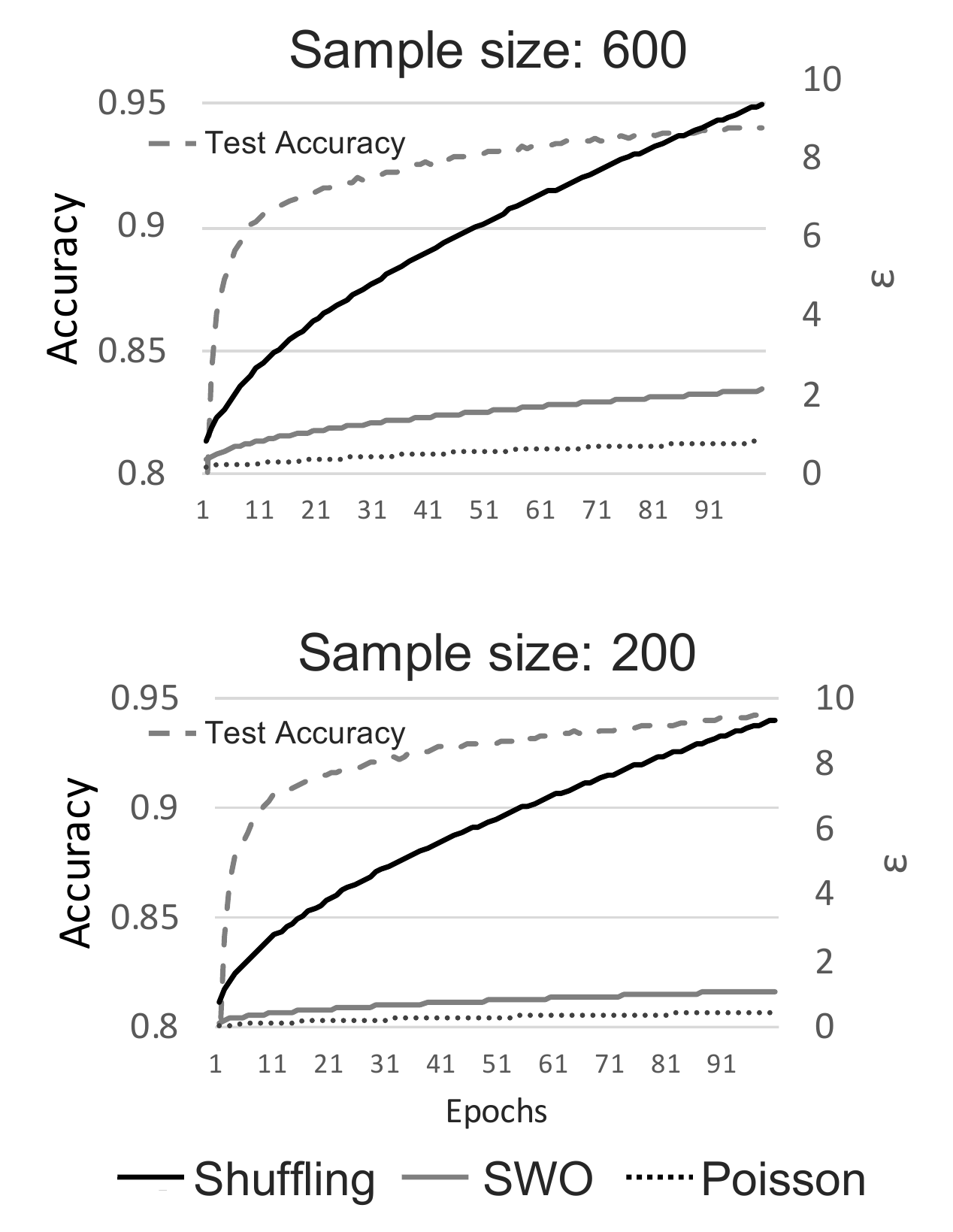}
	\caption{Accuracy and $\eps$ for MNIST over epochs for sample sizes 200 and 600.}
	\label{fig:eps}
		\vspace{-20pt}
\end{wrapfigure}
\textbf{CIFAR-10}
dataset consists of 50,000 training and 10,000 test color images classified into 10
classes~\cite{cifar}.
Each example is a $32\times32$ image with three channels (RGB).
We use the training setup from the TensorFlow tutorial~\cite{TFCifar10} for CIFAR-10
including the data augmentation step.
The same setup was also used in~\cite{abadi}.
The network consists of two convolutional layers followed by two fully connected
layers.
Similar to~\cite{abadi,dpmp_dl} we use a public dataset (CIFAR-100) to train a network with the same
architecture.
We then use the pre-trained network to train the fully connected layers using the CIFAR-10
dataset.
Each network is trained for 100 epochs with sample size of $\samplesize=2000$.
We use the same network setup as related work~\cite{abadi};
but better accuracy can be achieved with deeper networks.
The results for shuffling and sampling w/o replacement are
in Table~\ref{tbl:acc} (right). Similar to MNIST there is no significant difference
between the two.

\begin{table}[t]
\small
\caption{Test (Train) accuracy of
MNIST \& CIFAR10
models trained with samples
generated with $\shuffle$, $\poisson$
and sampling w/o replacement ($\swo$) and their differentially private (DP) variants
with incurred total~$\eps$.}\label{tbl:acc}
\centering
   % \hspace{50pt}
\makebox[0pt][c]{\parbox{1.2\textwidth}{%
    \begin{minipage}[b]{0.65\hsize}\centering
\begin{tabular}{r|c|c|c}
 & Shuffle & $\poisson$ & $\swo$ \\
 \hline
MNIST & {97.5 (98.33)} & {97.47 (98.31)} & {97.43 (98.31)} \\
DP MNIST & 94.06 (94.1) & 94.1 (94.01) & 94.03 (94.05) \\
\hline
$\eps$ & 9.39 & 0.82  & 2.13 \\
\end{tabular}
    \end{minipage}
    \hspace{-65pt}
    \begin{minipage}[b]{0.4\hsize}\centering
\begin{tabular}{r|c|c}
 & Shuffle & $\swo$ \\
 \hline
CIFAR-10 & 79.6 (83.2)  & 79 (82.9) \\
DP CIFAR-10 & 73.4 (72.3)  &  72.5 (71) \\
\hline
$\eps$ & 9.39  & 4.89 \\
\end{tabular}
    \end{minipage}
    \hfill
}}
\vspace{-10pt}
\end{table}

\paragraph{Sampling in differentially private training}
In Table~\ref{tbl:acc} (middle row) we compare the effect of sampling
approaches on DP training.
Similar to results reported in previous work DP training
degrades model performance.
However, accuracy between
sampling approaches is similar.
The difference between
the sampling mechanism is evident however
in the total privacy loss they occur.
The results in last row of Table~\ref{tbl:acc}
show that shuffling incurs the highest privacy
loss for the same number of epochs, in line with asymptotical guarantees in Table~\ref{tbl:epstheor}.
In Figure~\ref{fig:eps} we show that as expected smaller
sample (batch) size has a positive effect
on $\eps$ for sampling.

These results indicate that
if maintaining low privacy loss
is important then $\swo$ and $\poisson$  should be the preferred option
for obtaining batches:
sampling gives smaller privacy loss and same accuracy.

\ifFull
\begin{table}[htp]
\caption{Parameter $\epsilon$ of the differential private training in Table~\ref{tbl:acc}.
for sampling methods.}
\begin{center}
\begin{tabular}{r|c|c|c}
 & Shuffle & $\poisson$ & $\swo$ \\
 \hline
MNIST & 9.39 & 0.82  & 2.13 \\
CIFAR-10 & 9.38  & 1.68 & 4.89  \\
\end{tabular}
\end{center}
\label{tbl:dpeps}
\end{table}%
\fi

\section{Related work}
The use of TEEs
for privacy-preserving data analysis
has been considered in several prior works.
Multi-party machine learning
using Intel SGX and data-oblivious machine learning algorithms
has been described in~\cite{Ohrimenko2016}.
\textsc{Prochlo}~\cite{prochlo} shuffles user
records using TEEs
for anonymization. Secret shuffle allows \textsc{Prochlo} to obtain strong
guarantees from local DP algorithms~\cite{47557} that are applied
to records before the shuffle.
Systems in~\cite{opaque,Ohrimenko2015} consider map-reduce-like computation
for data analysis while hiding access
pattern between computations.
Slalom~\cite{slalom} proposes a way to partially outsource
inference to GPUs from TEEs while maintaining integrity and privacy.

Oblivious algorithms as software protection were first proposed in~\cite{Goldreich:1987:TTS:28395.28416, GoldreichO96}.
Recently, relaxation of security
guarantees for hiding memory accesses
have been considered in the context of differential privacy.
Allen~\textit{et al.}~\cite{DBLP:journals/corr/abs-1807-00736} propose an oblivious differentially-private framework for
designing DP algorithms that operate over data
that does not fit into private memory of a TEE (as opposed to sample-based analysis).
Chan~\textit{et al.}~\cite{Chan:2019:FDO:3310435.3310585} have considered implications
of relaxing the security guarantees of hiding memory accesses from data-oblivious definition to
the differentially-private variant.
Neither of these works looked at the problem of sampling.
In a concurrent and independent work, Shi~\cite{9152809} has developed a data-oblivious
sampling algorithm for generating one random sample from a stream of elements.
In comparison, we consider the setting where $\numsamples$ samples need to be generated
from a dataset of $n$ elements, thereby amortising the cost of processing the dataset across multiple
samples.

We refer the reader to~\cite{privacybook} for more information on differential privacy.
Besides work mentioned in~\S\ref{sec:dpsampling}, we highlight several
other works on the use of sampling for differential privacy.
Sample-Aggregate~\cite{Nissim:2007:SSS:1250790.1250803} 
is a framework based on sampling
where $\numsamples$ random samples are taken
such that in total all samples have $\approx n$ elements,
a function is evaluated on each sample, and $\numsamples$
outputs are then aggregated and reported with noise.
Kasiviswanathan~\emph{et al.}~\cite{DBLP:journals/siamcomp/KasiviswanathanLNRS11}
study concept classes that
can be learnt in differentially private manner based on a sample size and
number of interactions.
DP natural language models in~\cite{DBLP:conf/iclr/McMahanRT018} are trained
using a method of~\cite{abadi} while using data of a single user
as a mini-batch.
Amplification by sampling has been studied for
\Renyi~differential private mechanisms in~\cite{renyisampl}.
Finally, PINQ~\cite{McSherry:2009:PIQ:1559845.1559850},
assuming a trusted curator setting, describes a system for answering database queries
with DP guarantees.

\section*{Acknowledgements}
The authors would like to thank anonymous reviewers for useful feedback that helped improve the
paper. The authors are also grateful to Marc Brockschmidt, Jana Kulkarni, Sebastian Tschiatschek and
Santiago Zanella-B\'eguelin for insightful discussions on the topic of this work.

\bibliographystyle{abbrv}
\bibliography{bibfile,main}

\begin{thebibliography}{10}

\bibitem{cifar}
{CIFAR datasets}.
\newblock \url{http://www.cs.toronto.edu/~kriz/cifar.html}.
\newblock [Online; accessed 18-May-2019].

\bibitem{TFCifar10}
{TensorFlow Tutorial for CIFAR10 CNN}.
\newblock \url{
  https://github.com/tensorflow/models/tree/master/tutorials/image/cifar10},
  2019.
\newblock [Online; accessed 18-May-2019].

\bibitem{abadi}
M.~Abadi, A.~Chu, I.~Goodfellow, H.~B. McMahan, I.~Mironov, K.~Talwar, and
  L.~Zhang.
\newblock Deep learning with differential privacy.
\newblock In {\em Proceedings of the 2016 ACM SIGSAC Conference on Computer and
  Communications Security}, pages 308--318. ACM, 2016.

\bibitem{DBLP:journals/corr/abs-1807-00736}
J.~Allen, B.~Ding, J.~Kulkarni, H.~Nori, O.~Ohrimenko, and S.~Yekhanin.
\newblock An algorithmic framework for differentially private data analysis on
  trusted processors.
\newblock In {\em Conference on Neural Information Processing Systems
  (NeurIPS)}, 2019.

\bibitem{dptflink}
G.~Andrew, S.~Chien, and N.~Papernot.
\newblock {TensorFlow Privacy}.
\newblock \url{https://github.com/tensorflow/privacy}, 2019.
\newblock [Online; accessed 18-May-2019].

\bibitem{balle_subsampling}
B.~Balle, G.~Barthe, and M.~Gaboardi.
\newblock Privacy amplification by subsampling: Tight analyses via couplings
  and divergences.
\newblock In {\em Conference on Neural Information Processing Systems
  (NeurIPS)}, pages 6280--6290, 2018.

\bibitem{Bassily:2014:PER:2706700.2707412}
R.~Bassily, A.~D. Smith, and A.~Thakurta.
\newblock Private empirical risk minimization: Efficient algorithms and tight
  error bounds.
\newblock In {\em Symposium on Foundations of Computer Science (FOCS)}, 2014.

\bibitem{Beimel2014}
A.~Beimel, H.~Brenner, S.~P. Kasiviswanathan, and K.~Nissim.
\newblock Bounds on the sample complexity for private learning and private data
  release.
\newblock {\em Machine Learning}, 94(3):401--437, Mar 2014.

\bibitem{prochlo}
A.~Bittau, U.~Erlingsson, P.~Maniatis, I.~Mironov, A.~Raghunathan, D.~Lie,
  M.~Rudominer, U.~Kode, J.~Tinnes, and B.~Seefeld.
\newblock Prochlo: Strong privacy for analytics in the crowd.
\newblock In {\em ACM Symposium on Operating Systems Principles (SOSP)}, 2017.

\bibitem{Brasser:2017:SGE:3154768.3154779}
F.~Brasser, U.~M\"{u}ller, A.~Dmitrienko, K.~Kostiainen, S.~Capkun, and A.-R.
  Sadeghi.
\newblock Software grand exposure: {SGX} cache attacks are practical.
\newblock In {\em USENIX Workshop on Offensive Technologies (WOOT)}, 2017.

\bibitem{DBLP:journals/corr/abs-1802-08232}
N.~Carlini, C.~Liu, U.~Erlingsson, J.~Kos, and D.~Song.
\newblock {The Secret Sharer: Evaluating and Testing Unintended Memorization in
  Neural Networks}.
\newblock In {\em USENIX Security Symposium}, 2019.

\bibitem{Chan:2019:FDO:3310435.3310585}
T.-H.~H. Chan, K.-M. Chung, B.~M. Maggs, and E.~Shi.
\newblock Foundations of differentially oblivious algorithms.
\newblock In {\em ACM-SIAM Symposium on Discrete Algorithms (SODA)}, pages
  2448--2467, 2019.

\bibitem{privacybook}
C.~Dwork and A.~Roth.
\newblock The algorithmic foundations of differential privacy.
\newblock {\em Foundations and Trends in Theoretical Computer Science},
  9(3-4):211--407, 2014.

\bibitem{DBLP:conf/ccs/FredriksonJR15}
M.~Fredrikson, S.~Jha, and T.~Ristenpart.
\newblock Model inversion attacks that exploit confidence information and basic
  countermeasures.
\newblock In {\em ACM Conference on Computer and Communications Security
  (CCS)}, pages 1322--1333, 2015.

\bibitem{Goldreich:1987:TTS:28395.28416}
O.~Goldreich.
\newblock Towards a theory of software protection and simulation by oblivious
  rams.
\newblock In {\em ACM Symposium on Theory of Computing (STOC)}, 1987.

\bibitem{GoldreichO96}
O.~Goldreich and R.~Ostrovsky.
\newblock Software protection and simulation on oblivious {RAM}s.
\newblock {\em Journal of the ACM (JACM)}, 43(3), 1996.

\bibitem{DBLP:conf/stoc/Goodrich14}
M.~T. Goodrich.
\newblock Zig-zag sort: a simple deterministic data-oblivious sorting algorithm
  running in o(n log n) time.
\newblock In {\em ACM Symposium on Theory of Computing (STOC)}, pages 684--693,
  2014.

\bibitem{sgxcacheattacks}
J.~Gotzfried, M.~Eckert, S.~Schinzel, and T.~Muller.
\newblock Cache attacks on {Intel SGX}.
\newblock In {\em European Workshop on System Security (EuroSec)}, 2017.

\bibitem{sgx}
M.~Hoekstra, R.~Lal, P.~Pappachan, C.~Rozas, V.~Phegade, and J.~del Cuvillo.
\newblock Using innovative instructions to create trustworthy software
  solutions.
\newblock In {\em Workshop on Hardware and Architectural Support for Security
  and Privacy (HASP)}, 2013.

\bibitem{learnprivately}
S.~Kasiviswanathan, H.~Lee, K.~Nissim, S.~Raskhodnikova, and A.~Smith.
\newblock What can we learn privately?
\newblock {\em SIAM Journal on Computing}, 40(3):793--826, 2011.

\bibitem{DBLP:journals/siamcomp/KasiviswanathanLNRS11}
S.~P. Kasiviswanathan, H.~K. Lee, K.~Nissim, S.~Raskhodnikova, and A.~D. Smith.
\newblock What can we learn privately?
\newblock {\em {SIAM} J. Comput.}, 40(3):793--826, 2011.

\bibitem{Li:2012:SAD:2414456.2414474}
N.~Li, W.~Qardaji, and D.~Su.
\newblock On sampling, anonymization, and differential privacy or,
  k-anonymization meets differential privacy.
\newblock In {\em Proceedings of the ACM Symposium on Information, Computer and
  Communications Security}, ASIACCS, 2012.

\bibitem{DBLP:conf/iclr/McMahanRT018}
H.~B. McMahan, D.~Ramage, K.~Talwar, and L.~Zhang.
\newblock Learning differentially private recurrent language models.
\newblock In {\em International Conference on Learning Representations (ICLR)},
  2018.

\bibitem{McSherry:2009:PIQ:1559845.1559850}
F.~D. McSherry.
\newblock Privacy integrated queries: An extensible platform for
  privacy-preserving data analysis.
\newblock In {\em {SIGMOD}}, 2009.

\bibitem{DBLP:journals/corr/abs-1805-04049}
L.~Melis, C.~Song, E.~D. Cristofaro, and V.~Shmatikov.
\newblock Exploiting unintended feature leakage in collaborative learning.
\newblock In {\em IEEE Symposium on Security and Privacy (S\&P)}, 2019.

\bibitem{Nissim:2007:SSS:1250790.1250803}
K.~Nissim, S.~Raskhodnikova, and A.~Smith.
\newblock Smooth sensitivity and sampling in private data analysis.
\newblock In {\em ACM Symposium on Theory of Computing (STOC)}, 2007.

\bibitem{Ohrimenko2015}
O.~Ohrimenko, M.~Costa, C.~Fournet, C.~Gkantsidis, M.~Kohlweiss, and D.~Sharma.
\newblock Observing and preventing leakage in mapreduce.
\newblock In {\em ACM Conference on Computer and Communications Security
  (CCS)}, 2015.

\bibitem{melbshuffle}
O.~Ohrimenko, M.~T. Goodrich, R.~Tamassia, and E.~Upfal.
\newblock The {Melbourne} shuffle: Improving oblivious storage in the cloud.
\newblock In {\em International Colloquium on Automata, Languages and
  Programming (ICALP)}, volume 8573. Springer, 2014.

\bibitem{Ohrimenko2016}
O.~Ohrimenko, F.~Schuster, C.~Fournet, A.~Mehta, S.~Nowozin, K.~Vaswani, and
  M.~Costa.
\newblock Oblivious multi-party machine learning on trusted processors.
\newblock In {\em USENIX Security Symposium}, 2016.

\bibitem{Osvik2006}
D.~A. Osvik, A.~Shamir, and E.~Tromer.
\newblock Cache attacks and countermeasures: the case of {AES}.
\newblock In {\em RSA Conference Cryptographer's Track (CT-RSA)}, 2006.

\bibitem{DBLP:journals/corr/PatelPY17}
S.~Patel, G.~Persiano, and K.~Yeo.
\newblock Cacheshuffle: {A} family of oblivious shuffles.
\newblock In {\em International Colloquium on Automata, Languages and
  Programming (ICALP)}, 2018.

\bibitem{Riondato:2014:FAB:2556195.2556224}
M.~Riondato and E.~M. Kornaropoulos.
\newblock Fast approximation of betweenness centrality through sampling.
\newblock In {\em Proceedings of the 7th ACM International Conference on Web
  Search and Data Mining}, WSDM '14, pages 413--422, New York, NY, USA, 2014.
  ACM.

\bibitem{Riondato:2014:EDA:2663597.2629586}
M.~Riondato and E.~Upfal.
\newblock Efficient discovery of association rules and frequent itemsets
  through sampling with tight performance guarantees.
\newblock {\em ACM Trans. Knowl. Discov. Data}, 8(4):20:1--20:32, Aug. 2014.

\bibitem{9152809}
E.~{Shi}.
\newblock Path oblivious heap: Optimal and practical oblivious priority queue.
\newblock In {\em IEEE Symposium on Security and Privacy (S\&P)}, pages
  842--858, 2020.

\bibitem{DBLP:conf/sp/ShokriSSS17}
R.~Shokri, M.~Stronati, C.~Song, and V.~Shmatikov.
\newblock Membership inference attacks against machine learning models.
\newblock In {\em IEEE Symposium on Security and Privacy (S\&P)}, 2017.

\bibitem{6736861}
S.~{Song}, K.~{Chaudhuri}, and A.~D. {Sarwate}.
\newblock Stochastic gradient descent with differentially private updates.
\newblock In {\em 2013 IEEE Global Conference on Signal and Information
  Processing}, pages 245--248, Dec 2013.

\bibitem{pathoram}
E.~Stefanov, M.~van Dijk, E.~Shi, C.~W. Fletcher, L.~Ren, X.~Yu, and
  S.~Devadas.
\newblock Path {ORAM:} an extremely simple oblivious {RAM} protocol.
\newblock In {\em ACM Conference on Computer and Communications Security
  (CCS)}, 2013.

\bibitem{slalom}
F.~Tramer and D.~Boneh.
\newblock Slalom: Fast, verifiable and private execution of neural networks in
  trusted hardware.
\newblock In {\em International Conference on Learning Representations}, 2019.

\bibitem{Wang:2014:ODS:2660267.2660314}
X.~S. Wang, K.~Nayak, C.~Liu, T.-H.~H. Chan, E.~Shi, E.~Stefanov, and Y.~Huang.
\newblock Oblivious data structures.
\newblock In {\em ACM Conference on Computer and Communications Security
  (CCS)}, 2014.

\bibitem{renyisampl}
Y.~Wang, B.~Balle, and S.~Kasiviswanathan.
\newblock {Subsampled R{\'{e}}nyi Differential Privacy and Analytical Moments
  Accountant}.
\newblock In {\em Artificial Intelligence and Statistics Conference (AISTATS)},
  2019.

\bibitem{pmlr-v37-wangg15}
Y.-X. Wang, S.~Fienberg, and A.~Smola.
\newblock Privacy for free: Posterior sampling and stochastic gradient monte
  carlo.
\newblock In {\em International Conference on Machine Learning (ICML)}, 2015.

\bibitem{sgxsidechannels}
Y.~Xu, W.~Cui, and M.~Peinado.
\newblock Controlled-channel attacks: Deterministic side channels for untrusted
  operating systems.
\newblock In {\em IEEE Symposium on Security and Privacy (S\&P)}, 2015.

\bibitem{dpmp_dl}
L.~Yu, L.~Liu, C.~Pu, M.~Gursoy, and S.~Truex.
\newblock Differentially private model publishing for deep learning.
\newblock In {\em IEEE Symposium on Security and Privacy (S\&P)}, 2019.

\bibitem{opaque}
W.~Zheng, A.~Dave, J.~G. Beekman, R.~A. Popa, J.~E. Gonzalez, and I.~Stoica.
\newblock Opaque: An oblivious and encrypted distributed analytics platform.
\newblock In {\em USENIX Symposium on Networked Systems Design and
  Implementation (NSDI)}, 2017.

\bibitem{47557}
Úlfar Erlingsson, V.~Feldman, I.~Mironov, A.~Raghunathan, K.~Talwar, and
  A.~Thakurta.
\newblock Amplification by shuffling: From local to central differential
  privacy via anonymity.
\newblock In {\em ACM-SIAM Symposium on Discrete Algorithms (SODA)}, 2019.

\end{thebibliography}

\appendix

\clearpage
\section{Details of Oblivious $\swo$ Sampling (\S\ref{sec:oblswo})}
\label{app:swo}

\textbf{Sampling primitive}
A single sample of $\swo$ from domain $[1,n]$ can be instantiated using a
permutation~$\rho$ over $[1,n]$. The sample is defined
by elements that are mapped to the first $\samplesize$ elements, i.e., element $j$
is in the sample if $\rho(j) \le \samplesize$.
This procedure is described in Algorithm~\ref{alg:samplmemswo} for $\numsamples$ samples.
During the $\init$ call, random permutations are chosen (e.g.,
in the real implementation this would correspond to choosing
a random seed and then deriving $\numsamples$ seeds
for each permutation).
Then,
$\samplemember(i,j)$ returns $\True$
or $\False$ depending on whether
$j$ is in the $i$th sample or not.

Observe that each sample defined by the above primitive represents a valid $\swo$ sample.
Let $\sample_i$ be the sample that consists
of the first $\samplesize$ elements of the permutation $\rho_i$.
The probability of choosing a particular sample
is the probability of choosing one of the permutations where
the first $\samplesize$ elements are fixed.
Since there are $(n-\samplesize)!$ permutations
with first $\samplesize$ elements fixed:
the probability of $\sample_i$ is $(n-m)!/n!$
and is $\frac{1}{n}\frac{1}{n-1}\cdots\frac{1}{n-\samplesize+1}$ which is the probability of an $\swo$ sample.

\begin{algorithm}[h]
\begin{algorithmic}
	\caption[]{Instantiation of $\swo$ sampling for $\numsamples= n/\samplesize$ samples drawn from $\mathcal{F}_\swo^{n,\samplesize}$}
	\label{alg:samplmemswo}
\STATE $\init(n, \samplesize)$:
choose random permutations with domain $[1,n]$: $\rho_1, \rho_2, \ldots, \rho_\numsamples$
\STATE $\samplemember(i,j)$: If $\rho_i(j) \le \samplesize$ return $\True$, else $\False$
\end{algorithmic}
%\ifShort\vspace{-5pt}\fi
\end{algorithm}

\textbf{Security of Algorithm~\ref{alg:oblswo}}
The adversary observes an oblivious shuffle, a scan where an element is read and an encrypted pair is written,
another oblivious shuffle and then a scan that reveals the sample identifiers.
Since oblivious shuffle is independent of the content of $\cd$ and the shuffle permutation,
all patterns except for
revealing of the sample identifiers are
independent of the data.
We are left to argue that revealing sample ids and their locations (i.e., indices in the output $S$) does not reveal information
about the data nor the samples.
First note that there are $\samplesize$ copies of sample ids $1,2,\ldots,\numsamples$ associated with a ciphertext,
hence data-independent.
Second, note that locations of the revealed identifiers are random according
to the permutation chosen in the second shuffle step.
Since the permutation of the shuffles are hidden, the adversary does not learn the location of the tuple before and after the shuffle.

\begin{lemma}
\label{lemma:injectmap}
Let $\pi$ be a permutation over $n$ elements,
$\forall j \in[1,n], r_j \in [0,\numsamples]$ such that $\sum^n r_j = n$
and $\mathcal{K} = \{j~|~r_j \ge 1\}$.
For $l \in [1,n]$, let $\mapping(l) = \pi^{-1}(1 + \sum_{j=1}^{l-1} r_j)$.
Then $\mapping$ evaluated on keys in $\mathcal{K}$ is an injective random function
over $[1,n]$.
\end{lemma}
\begin{proof}
The statement follows from two observations: $\pi^{-1}$ is a permutation
and $\pi^{-1}$ is evaluated only on distinct elements from a set $[1,n]$.

The second observation is true since the mapping from $l$ to $1 + \sum_{j=1}^{l-1} r_j$,
when evaluated on $l\in[1,n]$, is injective
as it is strictly monotonic.
Moreover, $(1 + \sum_{j=1}^{l-1} r_j)\le n$
since $\sum^n r_j =  n$.

Co-domain of $\mapping$
appears independent of its input since
it is a subset of
the output of a random permutation function $\pi$
that has these properties by definition.
\end{proof}

\section{Details of Oblivious $\poisson$ Sampling (\S\ref{sec:oblpoi})}
\label{app:poi}

\textbf{Sampling primitive}
Instantiation of Poisson samples (Algorithm~\ref{alg:samplmemp}) is an extension of SWO sampling
that in addition also randomly chooses the size for each sample, $M_i$.
Recall that for SWO the sample size is fixed ($\samplesize$)
while Poisson sampling takes $\sampleprob$ as a parameter
and adds an element to the sample with probability $\sampleprob$.
Since the size of a Poisson sample is a random variable $\Binom(n, \sampleprob)$,
for each sample we draw a random variable
from this Binomial distribution and use it to determine
the sample size.

Observe that each sample defined by the above primitive represents a valid Poisson sample.
Let $\sample_i$ be the sample that consists
of the first $M_i$ elements of the permutation $\rho_i$.
The probability of choosing $\sample_i$
is the probability of the Binomial random variable being $M_i$
and then choosing a permutation where
the first $M_i$ elements are fixed.
The probability of the former is ${n \choose {M_i}} {\gamma^{M_i}}{(1-\gamma)}^{M_i}$.
Then the probability of $\sample_i$ is
${n \choose {M_i}} \gamma^{M_i}(1-\gamma)^{M_i} (n-M_i)!/n!
=\gamma^{M_i}(1-\gamma)^{M_i}/M_i!$
and is $\gamma^{M_i}(1-\gamma)^{M_i}$ if the element
order within the sample is not relevant.
Hence, samples produced by Algorithm~\ref{alg:samplmemp}
are distributed as $\poisson$ samples
of the corresponding sizes.

\samplmempalg

\textbf{Analysis}
Performance of oblivious Poisson sampling is dominated
by two oblivious shuffles and the non-oblivious sorting in Line~\ref{line:poisort}
since the replication scan is linear.

Security of Algorithm~\ref{alg:oblpoisson} follows that of $\swo$
except it requires $\numsamples'$ to be hidden
from an adversary. In particular, the adversary observes
two shuffles and a scan where one element is read and one written out
(hence, data-independent).
The content of the elements written out is encrypted or re-encrypted
to hide which elements
read from external memory are written out.
The total size of all the samples ($\cursize$) is protected
by padding the output to $n$. Hence, the adversary observes
only positions of every tuple (Line~\ref{line:posreveal}).
However, by construction position values are 1 to $n$ and are randomly shuffled,
hence, they are independent of the actual samples.
The sample boundary, which can be determined from the middle part of the tuples, $\enc(i)$,  in $\mathcal{S}$,
is encrypted and not revealed to the adversary.
Note that tuples with $\enc(0)$ denote padded dummy elements and do not belong to any sample.

\end{document}